\documentclass[a4paper,reqno,12pt]{amsart}

\usepackage{amsmath,amsthm,amssymb,amsfonts,color,graphicx}
\usepackage{amscd,verbatim}
\usepackage[all]{xy}

\theoremstyle{plain}
\newtheorem{theorem}{Theorem}

\theoremstyle{definition}

\theoremstyle{remark}

\numberwithin{equation}{section}
\numberwithin{theorem}{section}
\numberwithin{figure}{section}
\numberwithin{table}{section}

\newcommand{\A}{\ensuremath{\mathcal{A }}}
\newcommand{\B}{\ensuremath{\mathcal{B }}}
\newcommand{\C}{\ensuremath{\mathcal{C }}}
\newcommand{\D}{\ensuremath{\mathcal{D }}}

\newcommand{\F}{\ensuremath{\mathcal{F }}}
\newcommand{\HTf}{\ensuremath{\mathcal{H }}}
\newcommand{\I}{\ensuremath{\mathcal{I }}}
\newcommand{\La}{\ensuremath{\mathcal{L}}}
\newcommand{\Tf}{\ensuremath{\mathcal{T }}}

\newcommand{\KR}{\ensuremath{{\mathbb R }}}
\newcommand{\KZ}{\ensuremath{{\mathbb Z }}}

\newcommand{\T}{\ensuremath{{\mathbb T }}}
\newcommand{\X}{\ensuremath{\mathbb{X }}}

\newcommand{\Aut}{\operatorname{Aut}}

\newcommand{\Ind}{\operatorname{Ind}}

\newcommand{\Cpct}{\ensuremath{\mathcal{K }}}

\newcommand{\CrPr}[3]{ \ensuremath{#1\rtimes_{#3} {#2} }}

\begin{document}

\title{Topological T-duality and T-folds}

\author[P Bouwknegt]{Peter Bouwknegt}

\address[]{
Department of Theoretical Physics,
Research School of Physical Sciences and Engineering, and 
Department of Mathematics, Mathematical Sciences Institute, 
The Australian National University, 
Canberra, ACT 0200, Australia}
\email{peter.bouwknegt@anu.edu.au, asp105@rsphysse.anu.edu.au}

\author[A S Pande]{Ashwin S. Pande}

\begin{abstract}
We explicitly construct the $C^{\ast}-$algebras arising in the
formalism of Topological T-duality due to Mathai and Rosenberg
from string-theoretic data in several key examples. 
We construct a continuous-trace algebra with an action of
$\KR^d$ unique up to exterior equivalence
from the data of a smooth $\T^d$-equivariant gerbe on a trivial
bundle $X = W \times \T^d$. 
We argue that the `noncommutative T-duals'
of Mathai and Rosenberg \cite{MR2}, should be identified with the nongeometric
backgrounds well-known in string theory. We also argue that the
$C^{\ast}$-algebra $\CrPr{\A}{\KZ^d}{\alpha|_{\KZ^d}}$ should 
be identified with the T-folds of Hull \cite{Hull, HullTf}
which geometrize these backgrounds. 

We identify the charge group of D-branes
on T-fold backgrounds in the $C^{\ast}$-algebraic formalism of
Topological T-duality. We also study D-branes on T-fold backgrounds.
We show that the $K$-theory bundles of Ref.~\cite{herve} give a 
natural description of these objects.
\end{abstract}

\maketitle


\section{Introduction}

It has been known for a some time that string theories may be defined
on backgrounds which are not conventional geometries (See 
Refs. \cite{Hull, HullTf, Monodro, VM, SSN} and references therein).
We briefly review the construction of such backgrounds here following
Ref.~\cite{Hull}. Some of these non-geometric backgrounds may be 
constructed by considering 
`parametrized families' of string theories with the target space of
each string theory being a familiar space\footnote{We only consider
$\T^d$ as a target space here.} such as $\T^d$ or the ALE spaces.

String theories with target an $d$-dimensional torus possess
a `large' symmetry group (isomorphic to $\mathsf O(d,d,\KZ)$)
of which the `geometric' 
subgroup $\mathsf{GL}(d,\KZ) \subset \mathsf O(d,d,\KZ)$ is generated by large
diffeomorphisms of $\T^d$ while the rest are generated by T-dualities combined
with shifts of the Kalb-Ramond field. 

Now consider families of string theories parametrized by 
a space $W$, i.e.\ for each point of $W$ we consider a string theory with 
target $\T^d$.  It is clear that after moving around any topologically 
nontrivial loop in $W$, the theory must return to itself up to
a symmetry, i.e. we have a monodromy.
The continuous part of the symmetry may be gauged away and
hence the monodromy will take values in the 
large symmetry group of the theory on the $\T^d$ fiber.
If every monodromy lies in the geometric subgroup of the symmetry group then 
the parametrized family of theories is actually equivalent to a certain limit
of string theory with target a $\T^d$-bundle over $W$ with structure group 
the geometric group. 

If the monodromy does {\em not} lie in the geometric subgroup however, we 
cannot view the parametrized family as being defined by a theory on 
a geometric target spacetime. In this case it has been argued
that the target is (roughly speaking)
a $\T^d$ fibration over $W$ in which the fibers are
`glued together using T-dualities'. Such a space is not a geometry in the
ordinary sense of the term and has been termed a `T-fold' \cite{Hull,HullTf}
or a `monodrofold'  \cite{Monodro, VM}. 

Some T-folds arise as T-duals of
ordinary geometries, and hence string theory is certainly well defined on them,
even as a quantum theory. Conversely, Hull has shown that string theory 
on a general T-fold background may be defined as a theory with a 
sigma-model-like action with target a {\em geometric} background 
together with a set of constraints. He has argued that this action may
be quantised to obtain the full quantum theory associated with 
the T-fold \cite{Hull}.
This background is a (possibly non-principal) $\T^{2d}$ fibration over $W$ 
referred to in Ref.~\cite{HullTf} as the `correspondence space'.

We give a brief outline of the basic theory of T-folds here 
closely following Ref.~\cite{DBTF}, Sect.~2.
We consider a $\T^{2d}$ fibration over $W$. The formalism of 
Refs. \cite{DBTF,Hull} introduces a new indefinite metric field 
$L^{IJ}(Y)$ on this fibration. The physical spacetime is 
defined by choosing for each $y \in Y,$ a $\T^d \subseteq \T^{2d}$ which is 
null with respect to $L$. This defines a $\T^d$ fibration over $W$.
To be precise\footnote{See Ref.~\cite{DBTF} Sect. ~2.1.} 
we introduce 
projectors $\Pi^i_I,\tilde{\Pi}_{iI}$ $i=1,\ldots,n$, such that the matrix 
$$
\Pi = \left( \begin{array}{c}
\Pi^i_I \\ \tilde{\Pi}_{iI} 
\end{array} \right), i = 1,\ldots,n
$$
defines a choice of coordinates on each $\T^{2d}$ fiber in which 
$L$ has the form 
$$
\left(\begin{array}{cc}
\bf{0} & \bf{1} \\
\bf{1} & \bf{0} \\
\end{array} \right)
$$
then,
$L^{IJ} \Pi^i_I \Pi^j_J = L^{IJ} \tilde{\Pi}_{iI} \tilde{\Pi}_{jJ} = 0$.
Let $(\sigma^0,\sigma^1)$ be the coordinates on the worldsheet,
$\X^I, I = 1,\ldots,2d$, be coordinates on 
the $\T^{2d}$ fiber and $Y^A,A=1,\ldots, n$, coordinates on the base
$W$ of the fibration.  If we let
$X = \Pi \X$ and $\tilde{X} = \tilde{\Pi} \X$
then $X,Y$ are the coordinates on the physical spacetime and
$\tilde{X},Y$ are the coordinates on the T-dual spacetime.

The T-fold action is then
\begin{equation}
\La = -\frac{1}{2} \HTf_{IJ}(Y) \eta^{\alpha \beta} \partial_{\alpha}\X^I
\partial_{\beta}\X^J 
- \eta^{\alpha \beta} \I_{IA}(Y) \partial_{\alpha}\X^I
\partial_{\beta}Y^A + \La_N(Y).
\label{TFoldAction}
\end{equation}
Here $\X^I$ and $Y^A$ are worldsheet fields taking values in the base and
fiber of the fibration over $W$. Also, $\HTf$ is a family of metrics on the 
$\T^{2d}$ fiber parametrized by the base $W,$ and $\I$ is a `connection' for 
the $\T^{2d}$ fibration over $W$ (see Refs. \cite{DBTF,Hull} for details).
Also, $\La_N$ is the action for the fields with target the base $W$.
We have the explicit form 
\begin{gather}
\HTf = \left(\begin{array}{cc}
\HTf_{ij} & \HTf_i{}^j \\
\HTf^i{}_j & \HTf^{ij} \\
\end{array} \right) 
= \left(\begin{array}{cc}
G_{ij} - B_{ik} G^{kl}B_{lj} & B_{ik}G^{kj} \\
-G^{ik} B_{kj} & G^{ij} \\
\end{array} \right)\,,
\label{EqHBG}
\end{gather}
where $G$ is the metric on the $\T^{2d}$ fibers and $B$ is the
restriction of the $B$ field to the $\T^{2d}$ fibers (that is, we 
only consider the components of the metric and the $B$ field along 
the $\X$ coordinates).

The metric $G$ and $B$ field on the physical spacetime may be obtained from the
above data (Eqn. (4.9) of Ref. \cite{Hull}) as
\begin{align}
G^{ij} & = \Pi^{iI} \Pi^{jJ} \HTf_{IJ}\,, \nonumber \\
B_{ij} G^{jk} & = \tilde{\Pi}_i^I \Pi^{kJ} \HTf_{IJ}\,. 
\end{align}

Note that these are exactly the entries in the second column of the 
the matrix (Eqn.~\ref{EqHBG} above) for $\HTf$. 
It is clear that if we interchange $\Pi$ and $\tilde{\Pi},$ we will
obtain the first column of $\HTf$ above, which is exactly the T-dual
metric and $B$ field and would be associated to the T-dual spacetime. 

Thus the data of a T-fold (namely the metric $\HTf$ and the fixed choice
of projections $\Pi,\tilde{\Pi}$) determines two $\T^d$ fibrations over $W,$
each with their own metric and $B$ field,
one corresponding to the physical spacetime and the other to one of 
its T-duals. Exactly which T-dual is chosen is a matter of convention
and decided by the form of the metric $L$. In the conventions of
Ref. \cite{DBTF}, the T-fold is constructed from the original spacetime
and the one obtained by T-dualizing along all of the $\T$-orbits 
simultaneously. When both the original spacetime and the T-dual exist
as geometric backgrounds, the T-fold is the fiber product of
these two spacetimes, it is the `correspondence
space' (see below). On it, there is only the tensor field $\HTf$ 
which, as discussed above, contains the original metric and $B$ field.
However, T-folds continue to exist as
geometric spaces even when the original space or its full T-dual 
is non-geometric. An 
explicit example of this is the full T-dual $\T^3$ viewed as a
trivial $\T^2$-bundle over $\T$ with $H$-flux which
we will discuss in detail in Sect.~\ref{SecEg} below.

Now non-geometric backgrounds have also been seen in 
the $C^{\ast}$-algebraic formalism of Topological T-duality of 
Ref.~\cite{MR2}. In this paper
we study the relationship of these two formalisms to each other.
[See Ref.~\cite{SSN} for another viewpoint on this issue.]
In what follows we use the term `T-fold' to refer to the geometrization of
the nongeometric T-dual, i.e., to the (possibly non-principal) 
$\T^{2d}$ fibration over $W$ referred to as the `correspondence space'
in Ref.~\cite{HullTf}. This is the same as the correspondence space of
Refs.~\cite{MR1,MR2,BEM} (see below). 

We now outline the two formalisms that we will study in the following sections:
In the formalism of Topological T-duality of Ref.~\cite{MR2}, we begin with
a $C^{\ast}$-dynamical system $(\A,\alpha, \KR^d)$ 
such that the spectrum of $\A$ denoted by $\hat{\A}$ 
is a principal $\T^d$-bundle 
$p:X \to W$ and the $\KR^d$-action on $\A$ induces the given 
$\T^d$-action on $X=\hat{\A}$. 
Here $\hat{\A}$ is supposed to be a model for the topological type
of a target space with a $\T^d$-isometry. The cohomology class of the
$H$-flux on on the target space is modelled by the Dixmier-Douady class
of $\A$. The topological T-dual of $\A$ is \cite{MR2} the $C^{\ast}$-dynamical
system $(\CrPr{\A}{\hat{\KR}^d}{\alpha},\hat{\alpha},\hat{\KR}^d)$.
To each such $C^{\ast}$-dynamical system, there is associated
a function $f:W \to \T^{d(d-1)/2}$ termed the Mackey obstruction
map\footnote{See Sect.~\ref{SecCorr} below.} of the
system.  When this map is nullhomotopic (we say `there are no Mackey
obstructions') the spectrum $X^{\#}$ of this $C^{\ast}$-algebra is again a
principal circle bundle $q:X^{\#} \to W$ termed the T-dual bundle.
The authors of Ref.~\cite{MR2} also show that in this case 
there is a `correspondence space' 
homeomorphic to the fibered product
$X {\times}_W X^{\#}$ (the spectrum of $\CrPr{\A}{\KZ^d}{\alpha|_{\KZ^d}}$)
such that the following diagram of spaces commutes:

\begin{gather}
\xymatrix{ 
& X \times_W X^{\#} \ar[dr]_{r} \ar[dl]_{s} & \\
X \ar[dr]_{p} & & X^{\#} \ar[dl]_{q} \\
& W & \\ 
}
\label{CDMR}
\end{gather}
Here, $X,X^{\#}$ and $X \times_W X^{\#}$ are principal $\T^d$-bundles over $W$
with specified integral three-cohomology classes representing the $H$-flux, 
and all the maps are bundle projections. In addition, the pullbacks
of the $H$-fluxes along $r,s$ agree.
If there are Mackey obstructions however, the T-dual is in general a 
noncommutative space and the above diagram of spaces does not exist.
However, there is a similar diagram of $C^{\ast}$-algebras 
with the correspondence space replaced by the $C^{\ast}$-algebra
$\CrPr{\A}{\KZ^d}{\alpha|_{\KZ^d}}$.
We will argue in Sect.~\ref{SecCorr} below that this algebra is a natural
analogue of the T-folds of Hull when the T-dual is nongeometric.

The formalism of Ref.~\cite{HullTf} considers a smooth principal 
$\T^d$-bundle $p:X \to W$ (representing the target space)
together with a smooth $\T^d$-equivariant gerbe
with connection on $X$ (this is the conventional model for the $H$-flux). 
The $\T^d$-equivariance is needed in order to satisfy the conditions for 
gauging the sigma model on the target space (see Ref. \cite{Hull,HullTf}).
The authors define a certain obstruction cocycle 
$m_{\alpha \beta}^{IJ}$ which defines a class in $H^1(W,\KZ^m), m = d(d-1)/2$
with $I,J = 1,\ldots,d$, such that
the T-dual is geometric if and only if this class is zero. 
In this case, as shown in Ref.~\cite{HullTf}, we may define a smooth
correspondence space which gives a commutative diamond of spaces 
exactly as in Eqn.~\eqref{CDMR}. 
If some $m_{\alpha \beta}^{IJ}$ define nonzero cohomology classes in
$H^1(W,\KZ^m),$ a geometric T-dual does 
not exist. In this case, the correspondence space of 
Ref.~\cite{HullTf}, {\em is} a geometric space but it is a 
{\em non-principal} $\T^{2d}$-bundle over $W$. 
If we denote this correspondence space by $C$ here, we have a `diagram'
\begin{gather}
\xymatrix{ 
& C \ar@{.>}[dr] \ar[dl]_{s} & \\
X \ar[dr]_{p} & & \text{T-fold} \ar@{.>}[dl] \\
& W & \\ 
}
\label{CDHull}
\end{gather}
Here the dotted arrows are purely indicative, no quotient is implied.

Now the formalism of Ref.~\cite{HullTf} is defined
in the category of smooth manifolds while the formalism of
Ref.~\cite{MR2} is in a category of $C^{\ast}$-algebras.
A connection between these two seemingly different formalisms is of interest. 
Note that we do obtain the same correspondence space in both formalisms when, 
in one case, the Mackey obstruction vanishes and in the other, 
when all the $m^{IJ}$ are zero. 
It is natural to conjecture that the two obstructions are related.
We will show below that this is indeed the case. 
We will also construct the continuous-trace algebra associated with
certain types of $H$-flux on {\em trivial} torus bundles over a base
$W$ directly from the string theoretic data of Ref.~\cite{HullTf}.
In particular, we argue that the `non-commutative
T-duals' discussed in Ref.~\cite{MR2} should then be viewed as topological
approximations to the T-folds of Hull. 
We relate D-branes on the correspondence space to the formalism of
Topological T-duality in Sect.~\ref{SecDBCS} below.
We also compare our findings to the string-theoretic calculations
in Ref.~\cite{DBTF}.

\section{A Topological Approximation of $H$-flux\label{SecCorr}}

In this section we compare the formalisms of Mathai and Rosenberg \cite{MR1,MR2}
to the string theoretic T-dual studied by Belov, Minasian and 
Hull \cite{HullTf}. Note that the $C^{\ast}$-algebraic formalism
is only concerned with the topological aspects of T-duality, 
and hence a certain loss of information is to be expected on passing to 
it from string theory. 

We will see that the reason why the $C^{\ast}$-algebraic formulation 
of Topological T-duality agrees with string theory is due to the structure 
of the $H$-flux in string theory. It is well-known 
that the $H$-flux in string theory is required to satisfy the
WZW condition namely
$$
{\mathcal L}_{t} H = 0 \,,
$$
where $t$ is a tangent vector field on $X$ which is associated to the
torus action on $X$. As shown in Refs.~\cite{BouMa,Hull,HullTf}
this forces the $H$-flux to have the structure
$$
H = \frac{1}{6}\,p^{\ast}H_0^{IJK}\wedge \Theta_I \wedge \Theta_J 
\wedge \Theta_K + \frac{1}{2}\,p^{\ast}H_1^{IJ}\wedge \Theta_I \wedge \Theta_J 
+p^{\ast}H_2^I \wedge \Theta_I 
+p^{\ast}H_3 \,,
$$
where $\Theta_I, I=1,\ldots d$, are the components of the connection
form on $X$ and $H_k \in \Omega^k(W)$ (we use the notation of
Ref.~\cite{Hull} here, see Eqns.~(1.1), and (1.2) of that paper). 
We will see below that in several interesting cases 
this structure enables us to construct 
the $C^{\ast}$-algebraic T-dual.

In Ref.~\cite{Hull}, the authors consider a smooth $\T^d$-equivariant
gerbe with connection on $X$. In the case when the T-dual is
nongeometric, they show that the gerbe connection naturally gives
a $\T^{2d}$-bundle over $W$ with structure group $\mathsf{GL}(2d,\KZ)$.
They identify this with the correspondence space. 

Ref.~\cite{Hull}, Thm.~2.2 obtains an integral $1$-cocycle 
$m_{\alpha \beta} \in H^1(W,\wedge^2 \KZ^d)$
from a gerbe connection on $X$.
We will show below that the homotopy class of the Mackey obstruction map 
$f:W \to \T^{d(d-1)/2}$ of Ref.~\cite{MR2,herve} can be identified with 
the cohomology class of this cocycle. 

As noted in Ref.~\cite{Hull},
the class $m_{\alpha \beta}$ vanishes iff the T-dual is geometric and then the
correspondence space is the fibered product of the original space
and the T-dual. In the $C^{\ast}$-algebraic formalism, the spectrum of
$\Tf$ is exactly this fiber product when the Mackey obstruction vanishes.
Hence it seems natural to conjecture that $m_{\alpha \beta}$
is related to the Mackey obstruction.

Below, we will explicitly construct the Mackey obstruction map from the
data of a smooth equivariant gerbe and show that its homotopy class
is exactly $m_{\alpha \beta}$.
A heuristic argument for identifying the two is as follows:
In Ref.~\cite{Hull} it is shown that the component of 
the $H$-flux with two `legs' along
the torus fiber and one leg along the base determines 
classes\footnote{Ref.~\cite{Hull} after Eqn.~(2.20).}
$m_{\alpha \beta}^{IJ} \in H^1(W,\KZ)$.
Since $m^{IJ}$ is skew-symmetric in $I,J$ there will be
as many such classes as the number of elements in the basis of
$\wedge^2 \KZ^d$ determined by $I$. These give 
an element of $H^1(W,\wedge^2 \KZ^d)$.

In particular if $a^{IJ}$ is a vector-valued differential form 
representing $m_{\alpha \beta},$ the component of the $H$-flux associated
to $m_{\alpha \beta}$ is $a^{IJ} \wedge \Theta_I \wedge \Theta_J$,
where $\Theta_I$ is the connection form on $X$.
If we could construct a continuous-trace algebra $\A$ 
having Dixmier-Douady invariant $[H]$ together with
an $\KR^d$-action $\alpha$ covering the $\T^d$-action on $X$ 
thus determining a $C^{\ast}$-dynamical system $[\A,\alpha]$ on $X,$ 
then we would expect\footnote{In the notation of Ref.~\cite{MR2}, Thm.~2.3},
$m_{\alpha \beta}$ to be the cocycle representing $p_!(H)$, i.e., 
$m_{\alpha \beta} = p_! \circ F([\A,\alpha])$.
However, by the commutativity
of the diagram in Thm.~2.3 of that paper, 
$p_! \circ F = h \circ M$ and hence $m_{\alpha \beta}$ would be the
cohomology class representing the Mackey obstruction as well.

We now show that we can obtain the Mackey obstruction map itself from 
the string-theory calculation. 
We fix a basepoint $w_0$ in $W$ and the basepoint $(1,\ldots,1)$ in 
$\T^{d(d-1)/2}$.

\begin{theorem}
Given a $\T^d$-equivariant gerbe with connection on 
a smooth principal $\T^d$-bundle $p:X \to W$ 
\begin{enumerate}
\item If we fix the gauge of the gerbe connection on $X,$ 
there exists a smooth map $f:W \to \T^{d(d-1)/2}$ constructed from
the the gerbe connection. The map is natural under pullback of
gerbes. 
\item Under a gauge transformation of the gerbe connection, 
this map is multiplied by a phase $\gamma$ i.e., $f \to e^{i \gamma} f$.
\item This defines a unique based map 
$(W,w_0) \to (\T^{d(d-1)/2},(1,\ldots,1))$ 
such that the homotopy class of $f$ in $[W,\T^{d(d-1)/2}]$ is $p_!([H])$.
There is a geometric T-dual if and only if $f$ is nullhomotopic.
\end{enumerate}
\label{ThmMackey}
\end{theorem}

\begin{proof}
We use the notation of Ref.~\cite{HullTf} throughout this proof. 
We implicitly assume the results stated therein.
Let $U_{\alpha}$ be a good cover of $W$.
Let $X_{\alpha}= U_{\alpha} \times \T^d$ be a cover of $X$.
\begin{enumerate}
\item {}From Ref.~\cite{HullTf}, Eqn.~(2.4), we see that 
the curvature form of the gerbe connection on $X$ is
$$
H = \frac{1}{6}\,p^{\ast}H_0^{IJK} \Theta_I \wedge \Theta_J \wedge \Theta_K
+ \frac{1}{2}\,p^{\ast}H_1^{IJ}  \Theta_I \wedge \Theta_J
+ p^{\ast}H_2^{I}  \Theta_I + p^{\ast}H_3
$$
where $H_k$ are $k$-forms on $W$ and $\Theta_I, I=1,\ldots,d$,
are the coordinates on the $\T^d$ fiber.
Note that the $H$-flux is a {\em fixed} three-form and as such does
not change under gauge transformations. Hence, neither do the
$H_k$.  We take $H_0 = 0$ as we are working in the context of
the paper of Belov et al.~\cite{HullTf}. This implies that $H_1$ is a 
closed form on $W$. In each patch $X_{\alpha}$, the gerbe connection 
$B_{\alpha}$ (whose curvature is $H$ above) may be written as
\begin{equation} 
B_{\alpha} = B_{2\alpha} + B_{1\alpha}^I \wedge  \Theta_I
+ \frac{1}{2}\,B_{0 \alpha}^{IJ} \wedge  \Theta_I \wedge  \Theta_J
\label{EqBEquiv}
\end{equation}
where $B_{k \alpha}$ are $k$-forms on $X_{\alpha}$.
The equivariance condition which is required for the sigma model
to be well-defined\footnote{See Ref.~\cite{HullTf}.}
${\mathcal L}(\frac{\partial}{\partial \theta_I}) B_{\alpha} = 0$
implies that 
${\mathcal L}(\frac{\partial}{\partial \theta_I})B_{k \alpha} = 0, 
k=0,\ldots,3$. 
Hence, ${\mathcal L}(\frac{\partial}{\partial \theta_I}) B_{0 \alpha} = 0$
and so ${\partial B_{0 \alpha}}/{\partial \theta_I} = 0$. Hence,
$B_{0 \alpha}$ only depends on the coordinates of $W$ and not 
on the $\theta_I$. Thus, $B_{0 \alpha}$ is a well-defined function on 
$U_{\alpha} \subseteq W$.

We also have $B^{IJ}_{0 \alpha} - B^{IJ}_{0 \beta} = m^{IJ}_{\alpha \beta},$
(see Thm.~2.2 of Ref.~\cite{HullTf}),
that is, $m_{\alpha \beta}$ is an obstruction to $B_{0 \alpha}$ being
a global function on $W$.
Note that $m^{IJ}_{\alpha \beta}$ is the cocyle representing the
Mackey obstruction as discussed above.
Define $f^{IJ} = \exp(2 \pi i B^{IJ}_{0 \alpha}), 0 \leq I < J \leq d:$ 
It is clear that this construction is natural since, under pullback of 
gerbes, by definition, $B$ will map naturally.
The map $f$ is smooth since $B^{IJ}_{0\alpha}$ is smooth by definition.
\item Note that from the relation
$B^{IJ}_{0 \alpha} - B^{IJ}_{0 \beta} = m^{IJ}_{\alpha \beta},$
one might think that it would be possible to add a constant 
to each of the $B^{IJ}_{0 \alpha}$ without affecting $H$.

As we now show this procedure is a gerby gauge transformation of $B$. 
Given a gerbe on $X$ the gauge transformations are generated by a global line
bundle $q:K \to X$ with a connection $\nabla$ such that
under a gauge transformation $B_{\alpha} \rightarrow B_{\alpha} + 
\Omega_{\alpha}$ where $\Omega_{\alpha}$ is the restriction of the
curvature form $\Omega$ of $\nabla$ to $X_{\alpha}$. 
Now, $\Omega_{\alpha}$ may be written as 
$$
\Omega_{\alpha} = \frac{1}{2}\,\Omega^{IJ}_{0 \alpha} \Theta_I \wedge \Theta_J
+ \Omega^I_{1 \alpha}  \Theta_I + \Omega_{2 \alpha} 
$$
where $d \Omega_{\alpha} = 0$ since $\Omega$ is the curvature of
a line bundle. This implies that $d\Omega^{IJ}_{0 \alpha} = 0$.
Hence the $\Omega^{IJ}_{0 \alpha}$ are the restriction of a constant 
real-valued function on $X$ to $X_{\alpha}$.
Let the value of this function be $\gamma$.
Then, $f^{IJ} \rightarrow f^{IJ} \exp(i\gamma)$ under a gerby 
gauge transformation.

\item We choose the basepoint\footnote{This is not unusual, for example, 
in Ref.~\cite{herve}, Lemmas 6.5 and 6.6 the Mackey obstruction map is also 
required to be based.} $(1,\ldots,1)$ in $\T^{d(d-1)/2}$ 
and require that that the value of $f$ at $w_0$ should be the identity 
$(1,\ldots,1)$ of the Lie group $\T^{d(d-1)/2}$.  This fixes $f$ uniquely.

For $\mathsf G$ a group, let $\underline{\mathsf G}$ be the sheaf of  $\mathsf G$-valued functions
on $W$. Then we have a short exact sequence of sheaves 
$0 \to \underline{\KZ^m} \to \underline{\KR^m} \to \underline{\T^m} \to 0$.
The first few terms in the associated long exact sequence in sheaf
cohomology of $W$ are
$$
C(W,\KZ^m) \to C(W,\KR^m) \to C(W,\T^m) \overset{\beta}{\to} H^1(W,\KZ^m)
$$
where $\beta$ is the connecting map. In the special case $m=1$, and
$W=S^1$,  $\beta$ is the
map which sends a function to its winding number and hence may be identified
with the map $h$ of \cite{MR2}.
If we set $m = d (d-1)/2$
we see that $f^{IJ}$ defines an element of $C(W,\T^m)$ and its 
image under $\beta$ is precisely $m^{IJ}_{\alpha \beta}$.
In particular $\beta$ is the map which sends a function with
values in $\T^{d(d-1)/2}$ to the set of homotopy classes
$[W \to \T^{d(d-1)/2}] \cong H^1(W,\KZ^m)$.
Thus the homotopy class of $f$ is $m_{\alpha \beta}$ which is
the cohomology class of $H_1 = p_!(H)$.
If $f^{IJ}$ is nullhomotopic,
then $m_{\alpha \beta}$ is a coboundary and by Corollary 2.1 of
Ref.~\cite{HullTf} we see that the T-dual is a principal torus bundle.
The converse is also true by Corollary 2.3 of that reference.
\end{enumerate}
\end{proof}

Thus, we see from the above that the map $f$ has exactly the behaviour 
that the Mackey obstruction map should have (see the proof of Thm.~3.1
of Ref.~\cite{MR2}). In general the Mackey obstruction is only a
continuous map. However, note that the above construction will
always yield a {\em smooth} map. Thus, we suspect that the above construction
is not reversible, there is a loss of information when passing
from string theory to the formalism of Topological T-duality.

If we could naturally associate a continuous-trace algebra $\A$
with spectrum $X$ together with a $\KR^d$-action $\alpha$ to
the data of a smooth equivariant gerbe on $X$, then $f$ would
be the Mackey obstruction of $[\A,\alpha]$.  However, the construction
of such a pair directly from the gerbe is difficult in general.

We now restrict ourselves to the case when $X$ is a trivial
torus bundle since this case is well understood. We use the
work of Echterhoff, Nest and Oyono-Oyono \cite{herve} to compare
the results of Ref.~\cite{Hull} with $C^{\ast}$-algebraic topological
T-duality.  We further require that $H_0 = H_3 = 0$.

When $X = W \times \T^d$ and $f$ is not nullhomotopic, 
the $C^{\ast}$-algebraic formalism of Topological T-duality
(see Refs.~\cite{MR2, herve}) obtains a noncommutative principal
$\T^d$-bundle\footnote{Note that in Ref.~\cite{herve}, $f$ may
be nullhomotopic. Thus, their noncommutative bundles include ordinary
principal torus bundles as a subset. This fact will be implicitly used
below.} as the T-dual. 
We argue below that when $m_{\alpha \beta}$ 
is not a coboundary, the smooth gerbe formalism
of Belov et al.~\cite{HullTf} determines such a bundle:
\begin{theorem}
Suppose $X = W \times \T^d$ and suppose we are given
the data of a smooth $\T^d$-equivariant gerbe on $X$ in the sense of
Ref.~\cite{HullTf} with $H_0= H_3 =0$.
Let $\Cpct$ denote the $C^{\ast}$-algebra of compact operators on an
infinite-dimensional separable Hilbert space. 
\begin{enumerate}
\item The above data determines a $\KZ^d$-action $\theta$ on 
$C_0(W,\Cpct)$, where $\Cpct$ is the set of compact operators on a separable 
Hilbert space.
\item If we require that T-duals in the
$C^{\ast}$-algebraic formalism and the smooth equivariant gerbe
formalism determine the same principal $\T^d$-bundle when $H_1 = 0$,
then this $\KZ^d$-action is determined uniquely.
\item The canonical $\KZ^d$-action on $C_0(W,\Cpct)$ determines
a unique $C^{\ast}$-dynamical system $(\A,\alpha)$ such that $\hat{\A}=X$.
The Dixmier-Douady invariant of $\A$ equals the cohomology class
of $H$. The cohomology class of the Mackey obstruction of $(\A,\alpha)$ is 
equal to the class of the cocycle $m_{\alpha \beta}$ of Ref.~\cite{HullTf}.
\item The canonical $\KZ^d$-action on $C_0(W,\Cpct)$ determines
a unique noncommutative principal torus bundle over $W$.
This bundle is isomorphic to the $C^{\ast}$-algebraic T-dual
$\CrPr{\A}{\KR^d}{\alpha}$.
\end{enumerate}
\label{ThmZAction}
\end{theorem}
\begin{proof}
Since we have a trivial bundle the gerbe curvature form may be written as 
$$
H =  \frac12 H_1^{IJ} d \theta_I \wedge d \theta_J
+ H_2^I d \theta_I \,,
$$
where $H_1$ is a closed, integral form on $W$ and $\theta_I,I=1,\ldots,d$,
are coordinates on $\T^d$.
\begin{enumerate}
\item This is a consequence of Thm. ~2.3 of Ref.~\cite{herve}: The
exterior equivalence classes of $\KZ^d$-actions on $C_0(W,\Cpct)$ 
are parametrized by a group ${\mathcal E}_{\KZ^d}(W)$.  An element $\beta$ 
of this group is determined by two data:
\begin{itemize}
\item A Mackey obstruction map $f:W \to \T^{d(d-1)/2}$.
\item An isomorphism class of a principal $\T^d$ bundle $q:Y \to W$,  i.e.\ 
an element of $H^2(W,\KZ^d)$.
\end{itemize}
By Thm.~\ref{ThmMackey} above, we already have the map $f$. The 
principal bundle $q$ may also be obtained from the equivariant gerbe data: 
By Thm. ~2.2 of Ref.~\cite{HullTf}, the component $B_{1 \alpha}$ 
(see Eqn.~\eqref{EqBEquiv} above) of the gerbe connection $B_{\alpha}$ satisfies 
Eqn.~(2.17b) of Ref.~\cite{HullTf}
$$
B_{1\alpha} - B_{1\beta} = d \tilde{h}^I_{\alpha \beta}
+ m^{IJ}_{\alpha \beta} A_{\beta J}.
$$
where $\tilde{h}^I_{\alpha \beta}$ are real-valued functions on $W$
satisfying the cocycle identity\footnote{Note that since $X$ is a 
trivial bundle, its transition functions 
$\lambda_{\alpha \beta J}, J=1,\ldots,d$, may be taken to be zero in Thm. ~2.2
of Ref.~\cite{HullTf}.}
$
\tilde{h}^I_{\alpha \beta} + \tilde{h}^I_{\beta \gamma}
+ \tilde{h}^I_{\gamma \alpha} = 0.
$
Let $g^I_{\alpha \beta} = exp(2 \pi i \tilde{h}^I_{\alpha \beta})$.
The $g^I_{\alpha \beta}$ satisfy the cocycle condition 
$$
g^I_{\alpha \beta} g^I_{\beta \gamma} g^I_{\gamma \alpha} = 1
$$
on threefold intersections $X_{\alpha \beta \gamma}$ and hence 
define an element\footnote{Here $\underline{\T^d}$ 
is the sheaf of $\T^d$-valued functions on $W$ and the cohomology group is 
the sheaf cohomology.} of $H^1(W,\underline{\T^d})$. 
Hence they determine an isomorphism class of a principal $\T^d$-bundle 
over $W$.
Now the pair$[f^{IJ}, {g^I}_{\alpha \beta}]$ determine a 
$\KZ^{d}$-action $\theta$ on $C_0(W,\Cpct)$.

\item We need to show that the choice of $\tilde{h}_{\alpha \beta}$ is
canonical. After all, it is conceivable that one could
construct other cocycles on $W$ which might give a completely
different $\KZ^d$-action on $C_0(W,\Cpct)$.
However, the choice of $\tilde{h}_{\alpha \beta}$ is forced on us 
by the following argument:
Belov et al.\ show that\footnote{Corollary 2.1 of Ref.~\cite{HullTf}.}
when $m_{\alpha \beta}$ is a trivial cocycle, the $\tilde{h}_{\alpha \beta}$
determine the isomorphism class of the T-dual principal torus bundle. 
We therefore require that if $f$ constructed in Thm. ~\ref{ThmMackey} 
is nullhomotopic, the T-dual noncommutative principal $\T^d$-bundle obtained
above should reduce to a principal
$\T^d$-bundle (possibly with $H$-flux) whose characteristic class
should agree with that of the T-dual obtained from 
Belov et al.~Ref.~\cite{HullTf}.

It is a well-known result that when $f=0,$ i.e. $\beta$ has no
Mackey obstructions, the crossed product 
$\CrPr{C_0(W,\Cpct)}{\KZ^d}{\beta}$ 
is actually a continuous-trace
algebra on an ordinary principal $\T^d$-bundle.
In the $C^{\ast}$-algebraic
formalism of Topological T-duality, the algebra
$\CrPr{C_0(W,\Cpct)}{\KZ^d}{\beta}$  is strongly Morita equivalent
to the $C^{\ast}$-algebraic\footnote{ This is by Prop.~3.4 of Raeburn and 
Rosenberg \cite{RaeRos} with $G=\KR^d$ and $H=\KZ^d$.} 
T-dual $\CrPr{\A}{\KR^d}{\alpha}$ 
Thus, the choice of the cocycle $h$ above is canonical, as it is the
only one which makes $Y$ agree with the geometric T-dual 
in Ref.~\cite{HullTf}.
\item Let $\A = \Ind^{\KR^d}_{\KZ^d}(C_0(W,\Cpct),\theta)$
and $\alpha$ the induced $\KR^d$-action on $\A$.
This determines a $C^{\ast}$-dynamical system $(\A,\alpha)$ with
spectrum $X$ unique up to exterior equivalence. 
The Dixmier-Douady class of $\A$ is equal to the class of 
$H$ in integral cohomology by construction.
By the proof of Thm. ~\ref{ThmMackey} above, the homotopy class of
$f$ is equal to the cohomology class $m_{\alpha \beta}$ calculated from
the data of a $\T^d$-equivariant gerbe on $X$.
\item The canonical $\KZ^d$-action on $C_0(W,\Cpct)$ determines
a noncommutative principal $\T^d$-bundle over $W$ namely,
$\CrPr{C_0(W,\Cpct)}{\KZ^d}{\theta}$.  
This is strongly Morita equivalent to the $C^{\ast}$-algebraic T-dual
$\CrPr{\A}{\KR^d}{\alpha}$ by Prop. ~3.4 of Ref.~\cite{RaeRos}.
Since both algebras are stable, they are actually isomorphic.
\end{enumerate}
\end{proof}

Thus, we have constructed the $C^{\ast}$-algebraic T-dual\footnote{
The space that we describe as `the T-dual'
above is the one obtained by applying the T-duality transformation
along the all $d$ of the $S^1$-orbits
of $X$. It is one of the possible set of T-duals of the space $X$. 
Which T-dual we obtain is encoded in the choice of the action $\alpha$ above.
} 
from the string theoretic data on $X$.   It is thus natural to
view the noncommutative T-duals of Mathai and Rosenberg \cite{MR2}
as being topological approximations to the T-folds of Hull \cite{Hull}.

When the T-dual is geometric we may identify the correspondence space with 
the spectrum of $\Tf \cong \CrPr{\A}{\KZ^d}{\alpha|_{\KZ^d}}$. 
In light of the above theorem, it is natural to preserve this identification 
when the T-dual is non-geometric. 

Thus, we claim that $\Tf$ is the $C^{\ast}$-algebra naturally associated to the 
T-folds of Hull. It would be interesting
to construct this algebra directly from the data of the correspondence space
obtained by Belov et al.~\cite{HullTf} when the T-dual is nongeometric.

%
\section{D-Branes on the Correspondence Space \label{SecDBCS}}

In this section we make a few remarks on D-branes on T-folds and
their connection with the formalism of Topological T-duality.
We follow Ref.~\cite{DBTF} for the theory of D-branes on Hull's T-folds.

We would like to determine the charge group of the D-branes in this
background. To do this we consider the simplest case first, namely
the case when there are no Mackey obstructions so both spaces are
geometric. In this case, we have a commutative `diamond'
\begin{gather}
\xymatrix{ 
& X\times_W X^{\#} \ar[dr]_{r} \ar[dl]_{s} & \\
X \ar[dr]_{p} & & X^{\#} \ar[dl]_{q} \\
& W & \\ 
}
\label{CDBrane}
\end{gather}
[Note that the projection operators $\Pi$ and $1 - \Pi$
discussed by Hull \cite{HullTf} are obtained from the maps $r$ and $s$
restricted to each local coordinate patch.]
{}From Ref.~\cite{HullTf} Sect.~7, we see that in this situation, a
D-brane on the correspondence space should give D-branes on $X$ and $X^{\#}$.
Now due to the maps $r,s$ an element of the twisted $K$-homology of
the correspondence space
$K^{\tilde{H}}_{\ast}(X \times_W X^{\#})$, will
give rise to elements of $K^H_{\ast}(X)$ and $K^{H^{\#}}_{\ast}(X^{\#})$ which
are D-branes on $X$ and $X^{\#}$ (see for example Ref.~\cite{RS}).
Hence D-branes on the T-fold should be given by the elements of
$K^H_{\ast}(X \times_W X^{\#})$. Thus their charges should be
given by the twisted $K$-theory of the correspondence space
$K^{\ast}_H(X \times_W X^{\#})$ which, in this case,
is isomorphic to the operator-algebraic K-theory of $\Tf$. 
We conjecture that even when Mackey obstructions are present and
the spectrum of $\Tf$ is not a fibered product of two spaces, this 
identification continues to hold. Thus, we conjecture that the charge group of
D-branes on a T-fold should always be the K-theory of $\Tf$.
We now determine the $K$-theory of $\Tf:$

\begin{theorem}
Let $\B$ be a stable $C^{\ast}$-algebra with spectrum $W$. Let $\theta$
be a $\KZ^d$-action on $\B$. Consider the (stable) induced algebra
$\A = \Ind_{\KZ^d}^{\KR^d}(\B,\theta)$. 
Let the $\KR^d$-action on $\A$ be denoted $\alpha$
and define $\Tf = \CrPr{\A}{\KZ^d}{\alpha}$. 
\begin{enumerate} 
\item We have that
$$
K_0(\Tf) \cong K_1(\Tf) \cong K_0(\A)^d \oplus K_1(\A)^d\,.
$$
\item There is a natural automorphism $\phi_i:K_i(\Tf) \to K_{i}(\Tf), i=0,1,$
induced by the Connes-Thom isomorphism 
\end{enumerate}
\label{ThmIndKT}
\end{theorem}

\begin{proof} $\quad$ 
\begin{enumerate}
\item We consider $\A$ as possessing a natural $\KZ^d$ action 
(denoted $\alpha$ as well) obtained by restricting the $\KR^d$ action 
to the $\KZ^d$ subgroup. Consider $\D= \Ind_{\KZ^d}^{\KR^d}(\A)$. 
Let the $\KR^d$ action on $\D$ be denoted $\phi$.
We know that $\CrPr{\D}{\KR^d}{\phi}$ is strongly Morita equivalent to
$\CrPr{\A}{\KZ^d}{\alpha}$ by Thm.~2.2 of Ref.~\cite{RaeRos}.
Hence, by the Connes-Thom isomorphism theorem, it is enough to calculate
$K_{\ast}(\D)$. Now $\D$ was formed by applying the induced algebra
construction to the restriction of the $\KR^d$ action on 
$\Ind_{\KZ^d}^{\KR^d}(\B, \theta)$ to $\KZ^d \subseteq \KR^d$.
Thus, the induced algebra construction collapses and
$\D \cong \A \otimes C(\T^d)$. Hence, by the K\"unneth theorem 
$K_{\ast}(\D) \cong K_0(\A)^d \oplus K_1(\A)^d$.

\item The existence of the natural automorphism $\phi_i$ follows since
we may repeat the above construction with $\CrPr{\A}{\KR^d}{\alpha}$
instead of $\A$, note that there is an isomorphism
$\CrPr{\A}{\KZ^d}{\alpha} \cong 
\CrPr{ { (\CrPr{\A}{\KR^d}{\alpha})} }{\KZ^d}{\hat{\alpha}}$
and apply the Thom isomorphism theorem.
\end{enumerate}
\end{proof}

Note that the above proof its true for any induced algebra
$\Ind^{\KR^d}_{\KZ^d}(C_0(W,\Cpct))$. Such algebras
include those with noncommutative geometries as T-duals. Thus,
by our conjecture above, the charge group of D-branes on $\Tf$ will
always be $K_0(\A)^d \oplus K_1(\A)^d$.
It would be interesting to interpret this charge group physically
but a more detailed\footnote{It is not clear, for example, 
whether D-branes on T-folds would carry Chan-Paton bundles 
on their worldvolumes.} 
study of D-branes on T-fold backgrounds is needed. Still, we feel
that the above argument is natural enough to serve a preliminary step
towards such studies.

In the above theorem we could have begun with $\CrPr{\A}{\KR^d}{\alpha}$
instead of $\A$ and then the $K$-groups would be expressed in terms
of the $K$-theory of the crossed product. It would be interesting to
find an expression for the $K$-theory of $\Tf$ which does not
depend on the arbitrary choice of either $\A$ or $\CrPr{\A}{\KR^d}{\alpha}$.

Note that in the above theorem we could have replaced $\A$ with
a smooth subalgebra $\A^{\infty} \subseteq \A$ and still obtained
the same K-theory (by the Karoubi Density Theorem). 
In particular we could have taken the 
smooth subalgebra of smooth sections of the $\Cpct$-bundle over
$\hat{\A}$ which corresponds to $\A$. 
It might be interesting to try to define a smooth structure on $\A$ 
(in the sense of Connes) using the smooth structure of the gerbe on $X$.
However, as currently defined, $\A$ is only sensitive to topological
but not smooth information.

We recall that T-folds are obtained as geometrizations of parametrized families
of string theories on $\T^d$. Each $\T^d$ fiber is the background for
an exactly solvable string theory which could have D-brane excitations. 
Now consider families of these D-branes one in each fiber parametrized
by cycles in the base $W$. These would be obvious candidates for 
D-branes on T-folds. 
Such D-branes would wrap lifts of cycles on the base $W$ to the
correspondence space and also wrap the torus fiber. 
We note that in Ref.~\cite{DBTF} such
configurations have been studied from the string-theoretic point of
view in backgrounds which are geometrizations
of T-folds, i.e. in the correspondence space $\C$ of Ref.~\cite{HullTf}. 
We would like to determine the analogue of these types of D-brane
configurations in the formalism of Topological T-Duality. 

It is well known that the charge group of of D-branes on a space $X$ 
is the topological K-theory of $X$.
The D-brane configuration described above
would correspond to sections of a bundle of $K$-theory groups
over $W$ in the sense of Ref.~\cite{herve}. 

The authors of Ref.~\cite{herve} 
study $C^{\ast}$-bundles:  Roughly speaking, these are $C^{\ast}$-algebras 
whose elements are sections of bundles\footnote{These bundles do not need to be 
locally trivial.} of $C^{\ast}$-algebras over a base space $W$. 
For examples, the algebras $\A$, $\Tf$ and 
$\CrPr{\A}{\KR^d}{\alpha}$, discussed above, 
are all $C^{\ast}$-bundles over $W$.

To a $C^{\ast}$-bundle $\A(W)$ 
over a topological space $W$ satisfying certain conditions
(see Ref.~\cite{herve2} for details),
Ref.~\cite{herve} associates a bundle of abelian groups: 
The fiber of this bundle
over a point $x \in W$ is
the K-theory of the fiber of $\A(W)$  at that point. 
They also associate to this bundle a monodromy map: If $\A_w$ is
the fiber of $\A(W)$ over $w \in W,$ they define a map 
$\pi_1(W) \to \Aut(K_{\ast}(A_w))$. 
In addition the authors of Ref.~\cite{herve} show 
that the $K$-theory group of $\A(W)$ (for example, the charges of D-branes
on $\Tf$) may be obtained from
a spectral sequence whose $E^2$-term is a certain family of groups 
calculated using sections of the $K$-theory bundle of $\A(W)$.
This connection between parametrized families of D-branes and the
$K$-theory of $\Tf$ (which is the charge group of D-branes)
is interesting. 

Do such $K$-theory bundles exist for the $C^{\ast}$-algebra
$\Tf$ of the correspondence space? 
If so, continuous sections of the $K$-theory bundle of $\Tf$ would 
correspond to parametrized families of D-branes on the T-fold 
associated to $\Tf$. 
The following result holds for induced $C^{\ast}$-algebras $\A:$
\begin{theorem}
Let $\B = C_0(W,\Cpct)$ and let $\theta$
be a $\KZ^d$-action on $\B$. Consider the induced algebra
$\A = \Ind_{\KZ^d}^{\KR^d}(\B,\theta)$. 
Let the $\KR^d$-action on $\A$ be denoted $\alpha$
and define $\Tf = \CrPr{\A}{\KZ^d}{\alpha}$. 
\begin{enumerate}
\item $\Tf$ is a $C^{\ast}$ fibration over $W$ with 
fibers $\F \cong \CrPr{C(\T^{d},\Cpct)}{\KZ^d}{\alpha|_{\KZ^d}}$. 
\item It is a $K$ fibration in the sense of
Def.~4.1 of Ref.~\cite{herve} and hence possesses
a $K$-theory bundle and a monodromy map
$\rho: \pi_1(W) \to \Aut(K_{\ast}(\F))$.
\item When there is a geometric T-dual, the image of the monodromy map
lies in $\mathsf{GL}(2d,\KZ).$
\end{enumerate}
\label{ThmKFib}
\end{theorem}

\begin{proof} $\quad$
\begin{enumerate}
\item It is clear that $\Tf$ is a $C_0(W)$-algebra: Recall
the crossed product of a $C_0(W)$-algebra by a spectrum fixing
action of a group is also a $C_0(W)$-algebra. 
Now, $\Tf \cong \CrPr{\A^{\#}}{\KZ^d}{\alpha^{\#}|_{\KZ^d}}$ but
$\A^{\#}$ which is, by definition, $\CrPr{A}{\KR^d}{\alpha}$ is isomorphic to
$\CrPr{C_0(W,\Cpct)}{\KZ^d}{\theta}$ (since $\Tf$ is stable).
Thus, $\Tf$ is a $C^{\ast}$ fibration over $W$. Now $\Tf$ is 
by definition $\CrPr{\A}{\KZ^d}{\alpha}$. $\A$ is a $C^{\ast}$-bundle
over $W$ with fibers $C(\T^d,\Cpct)$. Since $\alpha|_{\KZ^d}$
is spectrum fixing, the fibers of the fibration $\Tf$ are
$\F \cong \CrPr{C(\T^d,\Cpct)}{\KZ^d}{\alpha|_{\KZ^d}}$.

\item By Lemma 8.4 of Ref.~\cite{herve}, $\A$ is a $K$ fibration.
Since $\Tf$ is the crossed product of $\A$ by a spectrum-fixing
$\KZ^d$-action, by Remark 2.3 Part (1) of Ref.~\cite{herve2}, 
$\Tf$ is a $K$ fibration as well.
The existence of a $K$-theory bundle and a monodromy map follow from
Prop. ~4.2 of Ref.~\cite{herve}. 

\item In the geometric case (i.e. when the Mackey obstruction vanishes),
the fibers of $\Tf$ over the base are just $C(\T^{2d},\Cpct)$. 
Since $K_i(C(\T^{2d},\Cpct)) \cong \KZ^{2d}, i = 0,1$, the image
of the monodromy map is $\Aut(\KZ^{2d}) \cong \mathsf{GL}(2d,\KZ)$.
\end{enumerate}
\end{proof}

Since the $K$-theory groups of
each fiber are discrete groups, if the $K$-theory bundle is trivial, the 
sections would be constant. This would correspond to choosing a
D-brane configuration on $W$ which restricts
to the {\em same} D-brane on each fiber. 
In general, if the bundle possesses a monodromy, there could be
nonconstant sections.
For example we could pick a D0-brane in every fiber and 
a T-dual D1-brane on applying the monodromy. 
Since the brane worldvolume jumps discontinuously, 
such D-brane configurations cannot exist when the background is geometric.
However, they might exist in non-geometric backgrounds.
In fact, these configurations have been observed in several physical examples
in Ref.~\cite{DBTF} where they are termed `generalized' D-branes.
In Ref.~\cite{DBTF}, parametrized families of D-branes for which
this does not occur are termed `geometric' D-branes.
The authors also consider generalized D-branes which `return to themselves'
only after (finitely) many traversals of such loops.

What would be the analogue of these in Topological T-duality? We suggest
the following construction:
As shown in Prop.~4.2 of Ref.~\cite{herve}
there is an action of $\pi_1(W)$
on the $K$-theory bundle of the algebra $\Tf$.  We can
define geometric families of D-branes in the $C^{\ast}$-algebraic formalism 
as constant sections of this bundle.
Non-geometric families would then correspond to nonconstant sections.
Recall that the worldvolume of a generalized D-brane
`returns to itself' after a finite number of circuits of the
cycle it wraps. This implies that the monodromy of the associated section 
of the $K$-theory bundle should be cyclic.
That is, in the formalism of Topological T-duality
generalized D-branes\footnote{ In the sense of Ref.~\cite{DBTF} }
would exist on $\Tf$ if at some point $w$ in $W$, the image of the monodromy 
map $\rho:\pi_1(W) \to \Aut(K_{\ast}(\A_w))$
was a finite cyclic subgroup of $\Aut(K_{\ast}(\A_w))$.
Obviously, when the monodromy map is trivial, all D-brane
configurations are geometric.


\section{Example: $\T^3$ with non T-dualizable $H$-flux \label{SecEg}}

In this section we study in detail the T-fold associated to a spacetime
which is a trivial $\T^2$-bundle over $\T$.
First, we make some remarks on the T-fold formalism applied to 
spacetimes which are trivial torus bundles over $W$.

Consider a space of the form $\T^d \times W$.  Let 
$y_i,i =1,\ldots,n$, be local coordinates on
$W$ and $\theta_k,k=1,\ldots,d$, be coordinates on the $\T^d$ fiber.
Assume that the metric on the space is the product metric of the
flat metric on $\T^d$ with the metric on $W$. The $B$ field
may then be written $B_{0\alpha}^{IJ} d\theta_I \wedge d \theta_j + \ldots$.
Here the dots denote terms in the $B$ field containing one or zero 
$d\theta_I$'s. It is clear that only $B_0$ will enter the expression for
$\HTf$ because only the $B_0$ component of the $B$ field is supported
completely on the torus fiber. The other components have nontrivial support
along the base $W$. Thus, the restriction of the $B$ field to the 
$\T^d$ fiber is locally $B_{0 \alpha}^{IJ}$.
If we examine the form for $\HTf$ we
see that, in the T-fold metric, the $B_{0 \alpha}^{IJ}$ appear along 
antidiagonal blocks. As noted in Thm. ~2.2 of Ref.~\cite{HullTf}, we have
$$
B_{0 \alpha}^{IJ} = B_{0 \beta}^{IJ} - m_{\alpha \beta}^{IJ}\,.
$$
Hence, a nonzero class $m_{\alpha \beta}$ 
causes the metric $\HTf$ to have nontrivial off-diagonal components. 
We claim that it is the twisting of the T-fold by these 
components of $\HTf$ that creates the monodromy in the total space of
the T-fold (as noted in Ref.~\cite{HullTf}).

Pick a open cover of $\T$ consisting of two open sets, with $U_1$ being
the complement of $0$ and $U_2$ the complement of its antipodal point.
Let $x$ be the coordinate
on the base $\T$ and $y,z$ be the coordinates on the 
$\T^2$ fiber. We pick the flat metric on $\T^3$ and an 
$H$-flux $H=N\, dx \wedge dy \wedge dz$.
We pick the $B$ field $B = Nx\, dy \wedge dz$ corresponding to this $H$-flux.

We recall that the connection forms on $\T^3$ are the $1$-forms
$dy,dz$. Therefore, comparing with Eqn.~\eqref{EqBEquiv} above, we see
that $B_{0\alpha} = Nx$. 
It is then clear that on changing charts, we have
that $B_{01} = B_{02} - N$. 
Now $m_{\alpha \beta}$ is an integral $2$-cocycle on $W$ and as such, 
is determined by specifying
integers for each pairwise intersection $U_{\alpha \beta}$. With the
given choice of charts, the cocycle $m_{\alpha \beta}$ has value 
$N$ on $U_1 \cap U_2$. Note that changing the coordinate functions on 
$U_{\alpha}$ will simply give a cohomologous cocycle.
In addition, using Thm. (\ref{ThmMackey}) in Sect. ~\ref{SecCorr}
above we see that the Mackey obstruction map is exactly the map
$x \to \exp(2 \pi i Nx)$. This is the Mackey obstruction map
constructed for an induced algebra with spectrum $\T^3$ by Mathai
and Rosenberg in Prop. ~4.1 of Ref.~\cite{MR2}. Note that the construction of
Ref.~\cite{MR2} exists for any continuous Mackey obstruction function. However,
we will always obtain a smooth function as a Mackey obstruction. The
reason, as mentioned in the introduction to Sect.~2,
is that Topological T-duality is a topological approximation 
(a continuous-trace algebra with a $\KR^d$ action) of a smooth 
$\T^d$ equivariant gerbe with connection.

Following Ref.~\cite{DBTF}, Sect.~2.3, let 
$$
\left(\begin{array}{c}
X^i \\ \tilde{X}^i \end{array}\right) 
= \left(\begin{array}{c} y \\ z \\ \tilde{y} \\ \tilde{z} \end{array} \right)
$$
be the coordinates on the corresponding T-fold.
The metric $\HTf$ on the T-fold 
is\footnote{See Eqn.~(2.20) of Ref.~\cite{DBTF}.}
\begin{equation}
\HTf = \left( \begin{array}{cccc}
1+(Nx)^2 & 0 & 0 & Nx \\
0 & 1+(Nx)^2 & -Nx & 0 \\
0 & -Nx & 1 & 0 \\
Nx & 0 & 0 & 1 \\
\end{array}\right).
\end{equation}
[Note that the metric is well defined even though the double T-dual
is not a geometric space. Note also that the metric and $B$ field on the
double T-dual may be recovered from $\HTf$ as discussed in the introduction.]

A detailed analysis of D-branes on T-fold backgrounds has been done
in Ref.~\cite{DBTF} Sect.~3.4. A D-brane on such a background is specified
by a projection valued function $\Pi_D$ on $W$ 
(a `Dirichlet projector') which picks out the directions
with Dirichlet boundary conditions on each $\T^{2d}$ fiber over $W$.
This also defines a Neumann projector $\Pi_N = ({\mathbf 1} - \Pi_D)$ on 
the T-fold which determines the Neumann directions on each fiber.
The D-brane boundary conditions are determined by 
$\Pi_D \partial_0 \X = 0$.  We can see that these determine two
sets of boundary conditions $\Pi \, \Pi_D \partial_0 \X = 0$ and
$\tilde{\Pi} \, \Pi_D \partial_0 \X = 0$.
Recall that the $\Pi \X$ were coordinates on $X$ and the $\tilde{\Pi} \X$
were coordinates on the T-dual $X^{\#}$. It is clear that the two sets
of boundary conditions above determine parametrized families of D-branes 
on $X$ and $X^{\#}$ respectively. These may not actually be D-branes as
the boundary conditions may possess a nontrivial monodromy. 
However, one can check that whenever a D-brane wrapping the T-fold
of $\T^3$-with $H$-flux is an allowed D-brane, the projections of
this brane to $X$ and $X^{\#}$ are also allowed D-branes.

We may use this phenomenon to justify the expression for $K_{i}(\Tf)$ 
that we have obtained in Thm. ~\ref{ThmIndKT} above, at least for
$\T^2$ fibrations.  Suppose $\A$ satisfies the hypotheses of
Thm.~\ref{ThmIndKT} above, then the spectrum of $\A$ is a
trivial $\T^2$-bundle over $W$, and the K-theory of the associated T-fold
is 
$$
K_i(\Tf) \cong (K_0(\A))^2 \oplus (K_1(\A))^2\,.
$$
Now, using the Thom isomorphism, this may be written
$$
K_i(\Tf) \cong K_0(\A) \oplus K_0(\CrPr{\A}{\KR^2}{\alpha})
\oplus K_1(\A) \oplus K_1(\CrPr{\A}{\KR^2}{\alpha})\,.
$$
We could interpret this as saying that
a D-brane on $\Tf$ would determine D-branes on $\A$ and 
$\CrPr{\A}{\KR^2}{\alpha}$ and further that the dimension of these
D-branes need not have the same parity as that of the D-brane
on $\Tf$.   A closer study of D-branes on $\Tf$ might give insight
into the formula for $K_i(\Tf)$ in the case $d \geq 2$.


\section{Conclusion}

In this paper we have argued that the T-folds in string theory
are naturally related to the noncommutative T-duals of
Ref.~\cite{MR2}: In a particular class of examples, 
the noncommutative T-duals of Ref.~\cite{MR2} may be naturally
obtained from string theoretic data.
It would be interesting to extend Thm.~\ref{ThmZAction} to the case
when $H_3 \neq 0$. One would have to study exterior equivalence classes
of $\KZ^d$-actions on continuous trace algebras on $W$ with nonzero
Dixmier-Douady invariant. It would also be interesting to 
extend the above theorem to nontrivial principal torus bundles.

We suspect that it should be possible to obtain the 
$C^{\ast}$-algebra $\Tf$ from
the data defining the non-principal $\T^{2d}$-bundle $\C$ which is the
correspondence space of Ref.~\cite{HullTf}. 

We have also argued that D-branes on T-folds have an
analogue in the $C^{\ast}$-algebraic formalism. It would be interesting
to pursue these analogies further.

\section{Acknowledgements}
We acknowledge financial support from the Australian Research Council 
through the Discovery Project `Generalized Geometries and their Applications'. 



\end{document}